\documentclass[10pt,twocolumn,conference]{IEEEtran}
\usepackage{amsmath,amsfonts,amssymb}
\usepackage{graphicx}
\usepackage{array}
\usepackage{subfigure}
\usepackage{cite}

\newtheorem{theorem}{Theorem}
\newtheorem{definition}{Definition}
\newtheorem{lemma}{Lemma}
\newtheorem{corollary}{Corollary}
\newtheorem{proposition}{Proposition}
\newtheorem{example}{Example}
\newtheorem{property}{Property}

\newcommand{\set}[1]{\mathcal{#1}}
\newcommand{\Real}{{\mathbb{R}}}

\def\setI{{\cal I}}

\def\GammaIn{\Gamma^{\text{In}}}
\def\setalpha{\alpha}
\def\N{{\cal N}}
\def\setIngleton{{\Delta}}

\title{The minimal set of Ingleton inequalities} 
\author{
\authorblockN{Laurent Guill\'{e}} 
\authorblockA{ENST Paris}
\and
\authorblockN{Terence H. Chan and Alex Grant}
\authorblockA{Institute for Telecommunications Research\\
University of South Australia}
}

\begin{document}
\maketitle

\begin{abstract}
  The Ingleton-LP bound is an outer bound for the multicast capacity
  region, assuming the use of linear network codes. Computation of the
  bound is performed on a polyhedral cone obtained by taking the
  intersection of half-spaces induced by the basic (Shannon-type)
  inequalities and Ingleton inequalities.  This paper simplifies the
  characterization of this cone, by obtaining the unique minimal set
  of Ingleton inequalities.  As a result, the effort required for
  computation of the Ingleton-LP bound can be greatly reduced.
\end{abstract}

\section{Introduction}
%

The network coding approach introduced
in~\cite{Ahlswede.Cai.ea00network,Li.Yeung.ea03linear} generalizes
routing by allowing intermediate nodes to forward packets that are
coded combinations of all received data packets.  
One fundamental problem in network coding is to understand the
capacity region and the classes of codes that achieve capacity. 

In the single session multicast scenario, the problem is well
understood -- the capacity region is characterized by max-flow/min-cut
bounds and linear network codes are sufficient to achieve maximal
throughput~\cite{Li.Yeung.ea03linear,Dougherty.Freiling.ea05insufficiency}.
Significant practical and theoretical complications arise in more
general multicast scenarios, involving more than one session.  

There are only a few tools in the literature for study of the capacity
region. One powerful theoretical tool bounds the capacity region by
the intersection of a set of hyperplanes (specified by the network
topology and connection requirement) and the set of entropy functions
$\Gamma^*$ (inner bound), or its closure $\bar{\Gamma}^*$ (outer
bound)~\cite{Song.Yeung.ea03zero-error,Yeung02first,Yeung.Li.ea06network}.
The set of entropy functions is formally introduced in Section
\ref{sec:two}.

In fact an exact expression for the capacity region has been obtained,
again in terms of $\Gamma^*$~\cite{YanYeu07}.
Unfortunately, the capacity region, or even the bounds cannot be
computed in practice, due to the lack of an explicit characterization
of the set of entropy functions for more than three random
variables. One way to resolve this difficulty is via relaxation of the
bound, replacing the set of entropy functions with the set of
polymatroids $\Gamma$. This yields the linear programming (LP)
bound~\cite{Yeung.Li.ea06network}.

The LP bound can however be quite loose. While the set of polymatroids
${\Gamma}^*$ is polyhedral, the set of entropy functions $\Gamma$ is
not~\cite{Matus07infinitely}. The addition of any finite number of
linear inequalities to the LP bound cannot tighten it to the capacity
region. Furthermore, the LP bound holds for any choice of network
codes (linear or non-linear). Hence, the LP bound can be even looser
when restricted to linear network codes.  To address the issue, a
modified LP bound was proposed in~\cite{Chan.Grant07dualities}.

The idea of the modified bound is quite simple. Given any network
code, the source messages and transmitted link messages are random
variables. Restriction to linear codes requires that the corresponding
entropy function satisfies the Ingleton inequality.  As a result, we
can tighten the LP bound for linear network codes by replacing
$\Gamma$ (the set of polymatroids) with $\GammaIn$ (a subset of
$\Gamma$ that satisfies all Ingleton inequalities).
 
Efficiently computation of this Ingleton-LP bound requires a compact
and explicit characterization of $\GammaIn$. As the set $\GammaIn$ is
the intersection of many half-spaces, one can greatly simplify the
characterization by identifying which inequalities (or half-spaces)
are redundant, meaning that they are implied by other
inequalities. The main objective of this paper is to understand the
relationship between these half-spaces, so as to simplify the
characterization of $\GammaIn$. Our main result, Theorem
\ref{thm:minimality} in Section~\ref{sec:three} is the identification
of the unique minimal set of Ingleton inequalities.  For reasons of
space, all simple proofs have been omitted, and longer proofs are
given in sketch form. 


\section{Entropy Functions}\label{sec:two}
Let $\N \triangleq \{ 1,2,\cdots, n\}$ induce a $2^{n-1}$-dimensional
real Euclidean space $\set{F}_n$ whose coordinates are indexed by the
set of all nonempty subsets $\alpha\subseteq\N$.  Each $h\in\set{F}_n$
is defined by $(h(\alpha) : \alpha \subseteq \N)$.  Although
$h(\alpha)$ is not defined for the empty set $\emptyset$, we will
assume $h(\emptyset) = 0$. Points $h\in\set{F}_n$ can also be
considered as functions $h:2^{\N}\mapsto\Real$.
\begin{definition}[Entropic function]\label{def:entropic}
  A function $h\in\set{F}_n$ is \emph{entropic} if there exists
  discrete random variables $X_1,\dots, X_n$ such that the joint
  entropy of $\{X_i: i\in\alpha\}$ is $h(\alpha)$ for all
  $\emptyset\neq\alpha\subseteq N$.  Furthermore, $h$ is \emph{almost
    entropic} if it is the limit of a sequence of entropic functions.
\end{definition}

Let $\Gamma^*_n$ be the set of all entropic functions. Its closure
$\bar{\Gamma}^*_n$ (i.e., the set of all almost entropic functions) is
well-known to be a closed, convex cone~\cite{Yeung97framework}.  An
important recent result with significant implications for
$\bar{\Gamma}^*_n$ is a series of linear information inequalities
obtained in~\cite{Matus07adhesivity}. Using this series,
$\bar\Gamma^*_n$ was proved to be non-polyhedral for $n\ge 4$. This
means that $\bar\Gamma^*_n$ cannot be defined by an intersection of
any finite number of linear information inequalities.

To simplify notation, set union will be denoted by concatenation, and
singletons and sets with one elements are not distinguished. For any
$\alpha,\beta \subseteq \N$ define
\begin{align*}
  h(\alpha |  \beta)&\triangleq h(\alpha \beta) - h(\beta)  \\
  I_h( \alpha ; \beta | \delta ) &\triangleq h(\alpha \delta) +
  h(\beta\delta) - h(\delta) - h(\alpha\beta\delta).
\end{align*}

\begin{proposition}\label{prop:basicineq}
  Let $h\in\bar\Gamma^*_n$. Then for all $\alpha,\beta,\delta
  \subseteq \N$,
  \begin{align}
    h(\alpha |  \beta) & \ge 0 \label{eq:condent}\\
    I_h( \alpha ; \beta | \delta ) & \ge 0. \label{eq:condinf}
  \end{align}
\end{proposition}

Proposition \ref{prop:basicineq} gives the \emph{basic inequalities},
namely the non-negativity of (conditional) entropy~(\ref{eq:condent})
and of (conditional) mutual information~(\ref{eq:condinf}).  This set
of basic inequalities is redundant, and the unique minimal set of
basic inequalities that yields all basic inequalities (as linear
combinations) is as follows
\begin{align}
  h( i  |  \N \backslash i) & \ge 0 \label{eq:elementalH}\\
  I_h( i  ; j | \delta ) & \ge 0  \label{eq:elementalI}
\end{align}
where $i\neq j\in\N$ and $\delta\subseteq\N \backslash \{i,j\}$.
Inequalities (\ref{eq:elementalH}) and (\ref{eq:elementalI}) are
called \emph{elemental} basic inequalities. See~\cite{Yeung02first}
for discussion of basic and elemental inequalities.

For any linear expression $\sum_{\alpha\subseteq \N} c_\alpha
h(\alpha)$, define the projection \emph{onto} a subset $\beta$ of $\N$
\begin{align}
  \sum_{\alpha\subseteq \N} c_\alpha h(\beta \cap \alpha).
\end{align}
Similarly, define the projection \emph{away from} $\beta$
\begin{align}
  \sum_{\alpha\subseteq \N} c_\alpha h( \alpha \backslash \beta).
\end{align}

Clearly, if two linear expressions $\sum_{\alpha\subseteq \N} c_\alpha
h(\alpha)$ and $\sum_{\alpha\subseteq \N} d_\alpha h(\alpha)$ are the
same (i.e, the two expressions map to the same real number for all
$h\in\set{F}_n$), then their projections onto (or away from) any
subset $\beta$ are identical.

\section{Ingleton inequalities:\\Properties and the minimal
  set}\label{sec:three}
\begin{definition}[Ingleton Inequality]
  An Ingleton inequality
  $J(h;\setalpha_1,\setalpha_2,\setalpha_3,\setalpha_4) \ge 0 $ is a
  linear inequality over $\set{F}_n $ defined in terms of four subsets
  $\setalpha_1,\setalpha_2,\setalpha_3,\setalpha_4$ of the ground set
  $\N$ where the Ingleton term
  $J(h;\setalpha_1,\setalpha_2,\setalpha_3,\setalpha_4)$ is defined as
\begin{align*}
  &h(\setalpha_1\setalpha_2)+h(\setalpha_1\setalpha_3)+h(\setalpha_1\setalpha_4)+h(\setalpha_2\setalpha_3)+h(\setalpha_2\setalpha_4)\\
  &\quad-h(\setalpha_1) - h(\setalpha_2) - h(\setalpha_3\setalpha_4) -
  h(\setalpha_1\setalpha_2\setalpha_3) -
  h(\setalpha_1\setalpha_2\setalpha_4).
\end{align*}
\end{definition}

{\it Remark:} Originally, the Ingleton inequality refers to the case
when $\setalpha_1,\dots, \setalpha_4$ are distinct subsets of
singletons. We extend its use to allow arbitrary subsets $\setalpha_1,
\cdots, \setalpha_4$ of $\N$.

Within the class of Ingleton inequalities indexed by four subsets of
$\N$, some are trivial inequalities while some can be derived from
basic inequalities or other Ingleton inequalities.  For example, if
$\setalpha_1 = \setalpha_2 = \setalpha_4 = \setalpha_4$, then $J(h,
\setalpha_1,\dots, \setalpha_4) = 0$ for all $h$ and the corresponding
Ingleton inequality is trivial.


We now list several properties of the Ingleton term, leading to the
minimal set of non-redundant Ingleton inequalities.

\begin{property}[Symmetry]\label{prop:symmetry}
\begin{align*}
J(h;\setalpha_1,\setalpha_2, \setalpha_3 ,\setalpha_4) & =  J(h;\setalpha_2,\setalpha_1, \setalpha_3 ,\setalpha_4) \\
& =  J(h;\setalpha_1,\setalpha_2, \setalpha_3 ,\setalpha_4) \\
& =  J(h;\setalpha_1,\setalpha_2, \setalpha_4 ,\setalpha_3).
\end{align*}
\end{property}
Thus exchanging $\setalpha_1$ with $\setalpha_2$ or $\setalpha_3$ with
$\setalpha_4$ does not change the value of the Ingleton term, and the
number of distinct Ingleton inequalities is at most $2^{4n-2}$.

%
\begin{property}[Extending basic inequalities]\label{prop:basic}
 \begin{align*}
   J(h;\setalpha_1,\setalpha_2,\emptyset,\setalpha_4)  &= I_h(\setalpha_1;\setalpha_2|\setalpha_4) \\
   J(h;\setalpha_1,\setalpha_1,\emptyset,\setalpha_2)  &= h(\setalpha_1| \setalpha_2).
 \end{align*}
\end{property}
Thus all basic inequalities are special cases of Ingleton inequalities
via proper selection of $\setalpha_1,\dots,
\setalpha_4$. Hence, $\GammaIn_n \subset\Gamma_n$.

\begin{property} \label{Part1} Let $\setalpha_1,\dots,
  \setalpha_4,\beta \subseteq \N$.  If $\beta \subseteq
  \setalpha_1\cap \setalpha_2$, then
 \begin{equation*}
   J(h;\setalpha_1, \setalpha_2, \setalpha_3, \setalpha_4) 
   = J(h;\setalpha_1, \setalpha_2, \setalpha_3 \beta , \setalpha_4
   \beta ) + h(\beta | \setalpha_3,\setalpha_4).
 \end{equation*}
 If $\beta \subseteq \setalpha_1\cap \setalpha_3$, then we have
\begin{multline*}
  J(h;\setalpha_1, \setalpha_2, \setalpha_3, \setalpha_4) \\
  = J(h;\setalpha_1, \setalpha_2 \beta, \setalpha_3, \setalpha_4
  \beta) + I_h(\beta;\setalpha_4|\setalpha_2).
\end{multline*}
On the other hand, if  $\beta \subseteq \setalpha_3\cap \setalpha_4$, then
\begin{multline*}
  J(h;\setalpha_1, \setalpha_2, \setalpha_3, \setalpha_4) \\
  = J(h;\setalpha_1 \beta, \setalpha_2 \beta, \setalpha_3,
  \setalpha_4) + I_h(\beta;\setalpha_2|\setalpha_1) +
  h(\beta|\setalpha_2 ).
\end{multline*}
\end{property}
To summarize, if an element appears in at least two subsets of an
Ingleton term, then we add that element to the remaining two subsets
without increasing the value of the Ingleton term.

\begin{property}\label{Part4}
  Let $a\subseteq \setalpha_2, b\subseteq \setalpha_3$ and $c\subseteq
  \setalpha_4$, then
\begin{multline*}
  J(h;abc, \setalpha_2, \setalpha_3,\setalpha_4) = I_h(\setalpha_3;\setalpha_4|abc)+I_h(\setalpha_3;c|\setalpha_2 a)   \\
  +I_h(\setalpha_4;b|\setalpha_2)+h(a|\setalpha_3\setalpha_4).
\end{multline*}
Similarly, if $a\subseteq \setalpha_1, b\subseteq \setalpha_2$ and
$c\subseteq \setalpha_3$, then
\begin{multline*}
J(h;\setalpha_1, \setalpha_2, \setalpha_3, abc)  
   = I_h(\setalpha_2;c|\setalpha_1b) + I_h(\setalpha_3;b|\setalpha_1)\\+I_h(\setalpha_3;a|\setalpha_2c)  
 + I_h(\setalpha_1;\setalpha_2|\setalpha_3ab)+h(c|\setalpha_2).
\end{multline*}
\end{property}
Consequently, if one of the subsets in the Ingleton term
$J(h;\setalpha_1, \setalpha_2, \setalpha_3,\setalpha_4)$ is contained
in the union of the other three subsets, then the corresponding
Ingleton inequality is implied by the basic inequalities. In fact, we
will prove that the converse is also true in the following theorem.
\begin{theorem}\label{polymatroid}
  An Ingleton inequality
  $J(h;\setalpha_1,\setalpha_2,\setalpha_3,\setalpha_4)\geq 0$ is
  implied by the basic inequalities if and only if one of the four
  subsets $\setalpha_1, \setalpha_2, \setalpha_3, \setalpha_4$ is
  contained in the union of the other three subsets.
\end{theorem}
\begin{proof}
  The {\it if}-part follows directly from Property~\ref{Part4}.  A
  sketch proof for the converse is given as follows.

  Suppose to the contrary that none of the four subsets
  $\setalpha_1,\setalpha_2,\setalpha_3,\setalpha_4$ are contained in
  the union of the remaining three subsets. To prove the converse, it
  suffices to show that there exists $h^* \in\set{F}_n$ satisfying
  every basic inequality but not the Ingleton inequality.
  To this end, write $\setalpha_i=\delta_i\cup \beta_i$, where (1)
  $|\delta_i|=1$, (2) $\beta_i\cap \delta_i =\emptyset$ and (3)
  $\emptyset\neq \delta_i \cap \setalpha_j=\emptyset$ for all $j\neq
  i$.  This term re-writing is possible by
  assumption. In~\cite{Zhang.Yeung98characterization}, an entropy
  function $g$ involving four elements was constructed which satisfies
  the basic inequalities but not the Ingleton inequality.  Since
  $J(h;\delta_1,\delta_2,\delta_3,\delta_4)$ involves only four
  elements, we can easily construct an $h^* \in \set{F}_n$ that
  satisfies all the basic inequalities and
 \begin{align*}
   J(h^*;\setalpha_1,\setalpha_2,\setalpha_3,\setalpha_4) & =
   J(h^*;\delta_1,\delta_2,\delta_3,\delta_4) \\ 
   & = J(g;\delta_1,\delta_2,\delta_3,\delta_4) < 0.
\end{align*}
Hence, $J(h^*;\delta_1,\delta_2,\delta_3,\delta_4) < 0$ and the theorem follows.
\end{proof}

Define the set of functions with non-negative Ingleton
term $$\GammaIn\triangleq\{h\in\set{F}_n:
J(h;\setalpha_1,\setalpha_2,\setalpha_3,\setalpha_4) \ge 0, \forall
\setalpha_i \subseteq \N\}.$$ It is clear that $\GammaIn$ is a closed
and convex cone. To characterize $\GammaIn$, we first define the
following sets of inequalities.
\begin{align}
\setIngleton_1 & =\{ J(h;i,j,\emptyset, \alpha)\ge 0  : \;   i,j \in \N, i\neq j,   \N \backslash \{i,j\} \} \\
\setIngleton_2 & =\{J(h; i, i,\emptyset, \N \backslash \{i\})  \ge 0 :\;  i\in \N \}.
\end{align}
By Property \ref{prop:basic}, $\setIngleton_1$ and $\setIngleton_2$
together imply all the basic inequalities.
 
\proposition\label{prop:suff} Let $\setalpha_1, \setalpha_2,
\setalpha_3, \setalpha_4$ be any four subsets of $\N $. Define
$\delta_1,\delta_2,\delta_3,\delta_4$ and $\beta$ as follows.
\begin{align}
\delta_i & \triangleq \alpha_i \backslash \bigcup_{j\neq i} \alpha_j \label{eqn:delta} \\
\beta &  \triangleq   \bigcup_i (\setalpha_i \backslash\delta_i ). \label{eqn:beta}
\end{align}
Then $\beta$ contains elements that appear in at least two subsets
$\setalpha_i$. Then, $J(h;\setalpha_1, \setalpha_2, \setalpha_3,
\setalpha_4) \ge J(h;\delta_1\beta, \delta_2\beta, \delta_3\beta,
\delta_4\beta)$.
\endproposition
\begin{proof}
  From Property \ref{Part1}, if an element appears in more than two
  subsets, then the element can be ``added'' to the remaining subsets
  without increasing the value of the Ingleton term. The result then
  follows.
\end{proof}

Let $\setIngleton_0$ be the set of Ingleton inequalities of the form
$J(h;\delta_1\beta,\delta_2\beta,\delta_3\beta,\delta_4\beta)\ge 0$
where $\delta_1, \delta_2, \delta_3, \delta_4$ and $ \beta$ are
disjoint subsets and $\delta_1, \delta_2, \delta_3, \delta_4$ are
nonempty. Furthermore, denote
$J(h;\delta_1\beta,\delta_2\beta,\delta_3\beta,\delta_4\beta)$ as $J(h
; \delta_1,\delta_2,\delta_3,\delta_4 | \beta)$.
\theorem\label{thm:suff} All Ingleton inequalities are implied by the
subset of Ingleton inequalities $\setIngleton =\setIngleton_0 \cup
\setIngleton_1 \cup \setIngleton_2$
\endtheorem
\begin{proof}
  If one subset $\alpha_i$ is contained in the union of the other
  three subsets, then by Theorem~\ref{polymatroid}, the inequality
  $J(h;\setalpha_1, \setalpha_2, \setalpha_3, \setalpha_4) \ge 0$ is
  implied by inequalities in $ \setIngleton_1 \cup
  \setIngleton_2$. Otherwise, by Proposition~\ref{prop:suff}, the
  inequality $J(h;\setalpha_1, \setalpha_2, \setalpha_3, \setalpha_4)
  \ge 0$ is implied by $J(h ; \delta_1,\delta_2,\delta_3,\delta_4 |
  \beta) \ge 0$ in $\setIngleton_0$ where $\beta, \delta_1 , \cdots,
  \delta_4$ are defined as in~(\ref{eqn:delta}) and~(\ref{eqn:beta}).
\end{proof}
Hence, $\GammaIn$ contains all $h\in\set{F}_n$ satisfying all
inequalities in $\setIngleton$. Additionally, the set is
full-dimensional, as shown below.

\proposition \label{Fulldim} There exists a function $h^* \in
\set{F}_n$ such that
$J(h^*;\setalpha_1,\setalpha_2,\setalpha_3,\setalpha_4) > 0$ for all
nontrivial inequalities (i.e., not always zero)
$J(h;\setalpha_1,\setalpha_2,\setalpha_3,\setalpha_4) \ge 0$.
\endproposition
\begin{proof} 
  Let $h^*(\alpha) \triangleq 2^n ( 1 - 2^{-|\alpha|})$. The
  proposition then follows by direct verification.
\end{proof}




\begin{corollary} \label{cor:one}
Suppose $c_i \ge 0$ for all $i\in\setI$. Then 
\begin{align}
  \sum_{i\in\setI} c_i J(h;\setalpha_1^i,\setalpha_2^i,\setalpha_3^i,\setalpha_4^i) = 0
\end{align}
for all $h\in\set{F}_n$ if and only if $c_i =0$ for all $i\in\setI$.
\end{corollary}
\begin{proof} 
  In the proof of Proposition~\ref{Fulldim}, we constructed $h$
  such that
  $J(h;\setalpha_1^i,\setalpha_2^i,\setalpha_3^i,\setalpha_4^i) \ge 1
  $ for all $i$. Therefore,
\begin{align}
  \sum_i c_i
  J(h;\setalpha_1^i,\setalpha_2^i,\setalpha_3^i,\setalpha_4^i) \ge
  \sum_i c_i.
\end{align}
Hence, if $\sum_i c_i
J(h;\setalpha_1^i,\setalpha_2^i,\setalpha_3^i,\setalpha_4^i) = 0$,
then $\sum_i c_i = 0$ or equivalently, $c_i=0$ for all $i$.
\end{proof}

By Proposition \ref{Fulldim},  we also have the following lemma.
\begin{lemma} \label{Lemmaforproof} Suppose that $\delta_1, \delta_2,
  \delta_3 , \delta_4, \beta $ are disjoint. Then
\begin{enumerate}
\item $J(h; \delta_1, \delta_2 , \delta_3 , \delta_4 | \beta) = 0$ if
  and only if (1) either $\delta_1$ or $\delta_2$ are empty, and (2)
  either $\delta_3$ or $\delta_4$ are empty.
\item $J(h ; \delta_1,\delta_2,\delta_3,\delta_4 | \beta) = J(h ;
  \delta_1^\prime ,\delta_2^\prime,\delta_3^\prime,\delta_4^\prime |
  \beta^\prime) $ if and only if (1) $\beta=\beta^\prime$, (2)
  $\{\delta_1, \delta_2\} = \{\delta_1^\prime, \delta_2^\prime\} $ and
  (3) $\{\delta_3, \delta_4\} = \{\delta_3^\prime, \delta_4^\prime\}
  $.
\end{enumerate}
\end{lemma}

So far, we have proved that the set of inequalities $\setIngleton$
implies all Ingleton inequalities and hence characterizes
$\GammaIn$. In the following, we will prove that $\setIngleton$ is
indeed the unique minimal set characterizing $\GammaIn$.



To obtain the minimal set of Ingleton inequalities, we need to
overcome an obstacle -- that two different choices of $\{\setalpha_i,
i=1, \cdots, 4\}$ might give the same Ingleton inequality.  Therefore,
the ``repeated'' inequalities must be
removed. 
\begin{example}\label{eg:first}
  Suppose that $\setalpha_1=\{1,5\}, \setalpha_2=\{2,5\},
  \setalpha_3=\{3,5\}$ and $\setalpha_4=\{4,5\}$.  If
  $\beta_i=\setalpha_i$ for $i = 1,2,3$ and $\beta_4 = \{4\}$, then
  $J(h; \setalpha_1, \cdots, \setalpha_4) = J(h; \beta_1, \cdots,
  \beta_4)$.
\end{example}
Fortunately, by our choice of $\setIngleton_1$ and $\setIngleton_2$,
no two inequalities are the same, and by Lemma~\ref{Lemmaforproof} ,
two inequalities are the same if and only the subsets in the Ingleton
term are permutations of each other as specified in
Lemma~\ref{Lemmaforproof}. Hence, those repeated inequalities can be
easily removed.  From now on, we assume that no inequalities in
$\setIngleton$ are the same by removing all these duplications.



\begin{theorem}\label{thm:minimality}
  No Ingleton inequality in $\setIngleton$ can be implied by others in
  $\setIngleton$. Consequently, the set of Ingleton inequalities
  $\setIngleton$ is the unique minimal set of Ingleton inequalities
  that characterizes $\GammaIn$.
\end{theorem}
\begin{proof}
  The proof for Theorem \ref{thm:minimality} when $n\le 5$ can be
  obtained by brute-force verification. When $n >5$, we will prove
  Theorem~\ref{thm:minimality} by considering several cases.  The
  proof is rather lengthy and will be given in the next section.
\end{proof}

By direct counting, the size of $\setIngleton$ can be shown to be
\begin{equation*}
n+ {n\choose 2} 2^{n-2} + \frac{1}{4}6^n - 5^n + \frac{3}{2}4^n - 3^n +\frac{1}{4}2^{n}.
\end{equation*}


\section{Proof of Theorem \ref{thm:minimality} }\label{sec:proof}
Suppose to the contrary of Theorem \ref{thm:minimality} that an
inequality $J(h;\setalpha_1,\setalpha_2,\setalpha_3,\setalpha_4)\ge 0$
in $\setIngleton$ is implied by other inequalities
$J(h;\setalpha_1^i,\setalpha_2^i,\setalpha_3^i,\setalpha_4^i)\ge 0$ in
$\setIngleton$.  By Farkas' Lemma (\cite{Schrijver2003}, p.61), there
exists non-negative constants $c_i$ such that
\begin{align}
J(h;\setalpha_1,\setalpha_2,\setalpha_3,\setalpha_4) &=
\sum_{i\in\setI} c_i
J(h;\setalpha_1^i,\setalpha_2^i,\setalpha_3^i,\setalpha_4^i) \\  
&=  \sum_{i\in \setI_0} c_i J(h;\delta_1^i , \delta_2^i , \delta_3^i , \delta_4^i| \beta^i) \nonumber \\
& \qquad  +  \sum_{i\in \setI_1} c_i J(h; k^i , l^i, \emptyset, \mu^i)  \nonumber \\
& \quad \qquad +   \sum_{i\in \setI_2} c_i J(h; {m^i},  {m^i}, \emptyset,  \N\backslash \{ m^i \}).\label{eqn:threecases}
\end{align}

To prove Theorem \ref{thm:minimality}, it suffices to show that $c_i =
0$ for all $i\in\setI \triangleq \setI_0 \cup \setI_1 \cup \setI_2$
and hence $ J(h;\setalpha_1,\setalpha_2,\setalpha_3,\setalpha_4) = 0$
contradicting the fact that $\setIngleton$ does not contain the
trivial inequality $0\geq 0$. We will prove
Theorem~\ref{thm:minimality} by considering three exhaustive cases,
{Case~A}, {Case~B} and {Case~C}.

\renewcommand{\thesubsectiondis}{Case \Alph{subsection}.}

\subsection{$J(h;\setalpha_1, \setalpha_2, \setalpha_3, \setalpha_4)\ge 0$ is in $\setIngleton_1$} 

In this case, the inequality $J(h;\setalpha_1, \setalpha_2,
\setalpha_3, \setalpha_4)\ge 0$ is of the form $J(h; i^0, j^0,
\emptyset, \mu^0 )\ge 0$ for some $i^0 \neq j^0 \in \N$ and $\mu
\subseteq \N \backslash \{i^0,j^0\}$.

Now, consider again the equality (\ref{eqn:threecases}) and project
both sides of it away from $i^0$. Then the left hand side becomes zero
by Lemma~\ref{Lemmaforproof}. Hence, we have
\begin{multline*} 
  0  =   \sum_{i\in \setI_0} c_i J(h; \delta_1^i \backslash i^0,\delta_2^i \backslash i^0,\delta_3^i \backslash i^0 ,\delta_4^i \backslash i^0 | \beta^i \backslash i^0) \\
    +  \sum_{i\in \setI_1} c_i J(h; k^i \backslash i^0 , l^i \backslash i^0, \emptyset, \mu^i \backslash i^0)   \\
   + \sum_{i\in \setI_2} c_i J(h; {m^i}\backslash i^0,
  {m^i}\backslash i^0, \emptyset, \N\backslash \{ m^i, i^0 \}).
\end{multline*}

By Corollary  \ref{cor:one}, 
\begin{equation*}
J(h; \delta_1^i \backslash i^0,\delta_2^i \backslash i^0,\delta_3^i \backslash i^0 ,\delta_4^i \backslash i^0 | \beta^i \backslash i^0) = 0
\end{equation*}
for all $i\in \setI_0$.  Then, by Lemma~\ref{Lemmaforproof}, we have
either $\delta_1^i $ or $\delta_2^i $ is empty, and either $\delta_3^i
$ or $\delta_4^i $ is empty.  However, this is impossible as
$\delta_1^i,\delta_2^i,\delta_3^i,\delta_4^i$ are disjoint and
nonempty.  Therefore, $c_i = 0$ for all $i\in\setI_0$.  Consequently,
\begin{multline*}
  J(h; i^0, j^0, \emptyset, \mu )
  =    \sum_{i\in I_1} c_i J(h; k^i , l^i, \emptyset, \mu^i)  \\
  + \sum_{i\in I_2} c_i J(h; {m^i}, {m^i}, \emptyset, \N\backslash \{
  m^i \}).
\end{multline*}
It is known that the elemental basic inequalities $\setIngleton_1 \cup
\setIngleton_2$ are not redundant~\cite{Yeung02first}.  The theorem is
thus proved in this case.

\subsection{$J(h;\setalpha_1, \setalpha_2, \setalpha_3,
  \setalpha_4)\ge 0$ is in $\setIngleton_2$} 

In this case, we can write the inequality as $J(h; i^0 , i^0 ,
\emptyset, \N\backslash \{i^0\} )\ge 0$ for some $i^0 \in \N$.  Again,
we can project both sides of the inequality away from ${i^0}$. Using
the same argument, we can conclude that $c_i = 0$ for
$i\in\setI_0$. Then the theorem again follows from that
$\setIngleton_1\cup \setIngleton_2$ is not redundant.

\subsection{$J(h;\setalpha_1, \setalpha_2, \setalpha_3, \setalpha_4)\ge 0$ is in $\setIngleton_0$}

In this case, $J(h;\setalpha_1, \setalpha_2, \setalpha_3,
\setalpha_4)\geq 0$ can be rewritten in the form $J(h;\delta_1,
\delta_2, \delta_3, \delta_4 | \beta)\ge 0$ for some disjoint subsets
$\delta_1, \delta_2, \delta_3, \delta_4,\beta$ such that $\delta_i$
are all nonempty. Again, assume that $J(h;\delta_1, \delta_2,
\delta_3, \delta_4|\beta) $ can be written as a linear combination of
other Ingleton inequalities as in (\ref{eqn:threecases}).

First, we will show that $c_i = 0$ for all $i\in  \setI_2$.
Let $h^*$ be the entropy function for random variables $\{X_1,\dots ,
X_n\}$ such that $h^*(\alpha) = H(X_i : i\in\alpha) = |\alpha|$.
Then, it is straightforward to prove that
\begin{align*}
J(h^*;\delta_1, \delta_2, \delta_3, \delta_4|\beta) & = 0, \\
J(h^*;\delta_1^i, \delta_2^i, \delta_3^i, \delta_4^i |\beta^i) & = 0  \\
J(h^* ; k^i , l^i, \emptyset, \mu^i) & = 0 .
\end{align*}
Substituting back into (\ref{eqn:threecases}), we have
\begin{equation*}
  0 = \sum_{i\in I_2} c_i. 
\end{equation*}
Consequently, $c_i = 0$ for all $i\in\setI_2$ and hence
\begin{multline}\label{caseC}
  J(h;\delta_1, \delta_2, \delta_3, \delta_4|\beta)  
   =   \sum_{i\in \setI_0} c_i J(h;\delta_1^i, \delta_2^i, \delta_3^i, \delta_4^i |\beta^i) \\
    +  \sum_{i\in \setI_1} c_i J(h; k^i , l^i, \emptyset, \mu^i)  .
\end{multline}
 
Our second task is to show that $c_i = 0$ for all $i\in \setI_1$.
Again, we will use a similar projection trick.  Consider any $i^0 \in
\setI_1$ and the corresponding inequality
\begin{equation*}
  J(h; k^{i_0} , l^{i_0}, \emptyset, \mu^{i_0})  \ge 0.
\end{equation*}

We can project both sides of (\ref{caseC}) onto $k^{i_0}$ and
$l^{i_0}$. Clearly, the right hand side contains the term $c_{i_0}
J(h; k^{i_0} , l^{i_0}, \emptyset, \emptyset) $.  Thus, the left hand
side after projection cannot be zero.  As a result, (1) $\{ k^{i_0},
l^{i_0}\} \cap (\delta_1\delta_2)$ and $\{ k^{i_0}, l^{i_0} \} \cap
(\delta_3 \delta_4)$ are nonempty, and (2) $k^{i_0}$ and $l^{i_0}$ are
not in the same subset.

Therefore, we may assume without loss of generality that $k^{i_0} \in
\delta_1$ and that $l^{i_0} \in \delta_3$.  Since $\delta_2$ and
$\delta_4$ are nonempty, we can pick $a\in \delta_2$ and
$b\in\delta_4$. Then we can project (\ref{caseC}) onto $\{ k^{i_0},
l^{i_0} , a, b \}$.
After projection, the left hand side becomes
\begin{equation*}
  J(h;  k^{i_0}, a,  l^{i_0} , b  )
\end{equation*}
and the right hand side is a summation of several Ingleton
inequalities (involving at most four variables) including $ c_{i_0}
J(h; k^{i_0} , l^{i_0}, \emptyset, \mu^i)$.  As
Theorem~\ref{thm:minimality} holds when $n=4$, we have $c_{i_0} = 0$.
Repeating the same argument for all $i\in \setI_1$, we prove that $c_i =
0$ for all $i\in\setI_1$.

Now (\ref{caseC}) can  be rewritten as  
 \begin{equation}\label{caseCone}
J(h;\delta_1 , \delta_2, \delta_3, \delta_4 | \beta)  =   \sum_{i\in \setI_0} c_i J(h;\delta_1^i\beta^i, \delta_2^i\beta^i, \delta_3^i\beta^i, \delta_4^i\beta^i) 
\end{equation}

Assume that $c_i >0 $ for all $i\in\setI_0$ in (\ref{caseCone}).  Now,
to prove Theorem~\ref{thm:minimality}, it suffices to prove the
following statement.

\begin{proposition}[Induction Hypothesis ${\cal H}(n)$]
  Let $n$ be the number of set elements involved in the left hand side
  of the expression in (\ref{caseCone}).  Suppose that the equality
  (\ref{caseCone}) holds. Then
  \begin{align*}
    J(h;\delta_1 , \delta_2 , \delta_3 , \delta_4 | \beta) =
    J(h;\delta_1^i , \delta_2^i , \delta_3^i, \delta_4^i | \beta^i) .
  \end{align*}
\end{proposition}
We have verified cases up to $n \le 5 $. The case when $n \ge 6 $ will
be proved by induction.  To this end, we first prove the claim that
any element appearing in the right hand side of (\ref{caseCone}) must
also appear in the left hand side.

Suppose to the contrary that (1) the equality (\ref{caseCone}) holds
and (2) there exists an element $a$ appearing only on the right hand
side. Further suppose that $a \in \delta^{i_0}_j$ for some $i^0\in
\setI_0$.  Then it is easy to find another element $b$ such that after
projection onto $\{a,b\}$, the right hand side of~(\ref{caseCone}) has
a term $I_h(a,b)$.  However, the left hand side of~(\ref{caseCone})
can be shown to be zero, contradicting to Corollary~\ref{cor:one}

On the other hand, if $a\in \beta^{i_0}$, then we can pick elements $b
\in \delta_1^{i_0}$ and $c \in \delta_2^{i_0}$. Projecting both sides
of~(\ref{caseCone}) onto $\{a,b,c\}$, the left hand side is either
zero or $I_h(b;c)$, while the right hand side of (\ref{caseCone}) is
nonzero as it contains a term $c_{i_0} I_h(b; c | a )$. Contradiction
occurs and hence all elements appearing in the right hand side will
also appear in the left hand side.

Suppose that the induction hypothesis ${\cal H}(n)$ holds. We now aim
to prove that ${\cal H}(n+1)$ also holds.  Suppose~(\ref{caseCone})
involves at most $n+1$ set elements where $n \ge 5$. We consider two
sub-cases, \emph{C.1} and \emph{C.2}.

\renewcommand{\thesubsubsectiondis}{\thesubsectiondis\arabic{subsubsection}.}

\subsubsection{$|\beta|\ge 2$}
Let $a, b \in \beta$ and $a\neq b$.  Projecting both sides of
(\ref{caseCone}) away from $a$. Then (\ref{caseCone}) becomes
\begin{multline}\label{caseA}
  J(h;\delta_1 , \delta_2 , \delta_3 , \delta_4 | \beta \backslash a)
  \\ = \sum_{i\in\setI_0} c_i J(h;\delta_1^i\backslash a ,
  \delta_2^i\backslash a , \delta_3^i\backslash a,
  \delta_4^i\backslash a | \beta^i \backslash a).
\end{multline}
As (\ref{caseA}) involves only $n$ variables, applying the induction
hypothesis, we have $\beta_i = \beta \backslash a$ or $\beta_i =
\beta$.  Similarly, we can prove that $\beta_i = \beta \backslash b$
or $\beta_i = \beta$.  Consequently, $\beta_i = \beta$ for all
$i\in\setI_0$.  Again, from (\ref{caseA}),
\begin{equation*}
J(h;\delta_1 , \delta_2 , \delta_3 , \delta_4 | \beta \backslash a) 
 = J(h;\delta_1^i \backslash a , \delta_2^i\backslash a, \delta_3^i\backslash a, \delta_4^i\backslash a | \beta^i).
\end{equation*}
By Lemma \ref{Lemmaforproof} and the induction hypothesis, we have
$\{\delta_1 , \delta_2 \} = \{\delta_1^i , \delta_2^i \}$ and
$\{\delta_3 , \delta_4 \} = \{\delta_3^i , \delta_4^i \}$.  The
hypothesis ${\cal H}(n+1)$ then holds.

\subsubsection{$|\beta |\le 1$}

Since $n+1 \ge 6$ and $|\beta| \le 1$, there exists distinct $a$ and
$b$ in a subset $\delta_i$ for some $i$.  Assume without loss of
generality that $i=1$.  Projecting (\ref{caseCone}) away from $a$,
\begin{multline*}
  J(h; \delta_1 \backslash a, \delta_2, \delta_3, \delta_4 | \beta) \\
  = \sum_{i\in \setI_0} c_i J(h;\delta_1^i\backslash a ,
  \delta_2^i\backslash a , \delta_3^i \backslash a,
  \delta_4^i\backslash a | \beta^i \backslash a) .
\end{multline*} 
By Lemma \ref{Lemmaforproof} and the induction hypothesis, $\beta^i =
\beta$ or $\beta^i = \beta \cup \{a\}$.  Similarly, by projecting
(\ref{caseCone}) away from $b$, we have $\beta^i = \beta$ or $\beta^i
= \beta \cup \{b\}$.  Consequently, $\beta^i = \beta$ for all
$i\in\setI_0$.  Thus (\ref{caseA}) becomes
\begin{equation*} 
 J(h;\delta_1 , \delta_2 , \delta_3 , \delta_4 | \beta)   =   \sum_{i\in \setI_0} c_i  J(h;\delta_1^i  , \delta_2^i , \delta_3^i, \delta_4^i | \beta).
\end{equation*}
and hence after projection away from $a$, 
\begin{multline} \label{finala}
  J(h; \delta_1 \backslash a, \delta_2, \delta_3, \delta_4 | \beta)  \\
  = \sum_{i\in \setI_0} c_i J(h;\delta_1^i\backslash a ,
  \delta_2^i\backslash a , \delta_3^i \backslash a,
  \delta_4^i\backslash a | \beta ) .
\end{multline}
 Similarly, projecting (\ref{caseCone}) away from $b$, we have
\begin{multline}\label{finalb}
  J(h; \delta_1 \backslash b, \delta_2, \delta_3, \delta_4 | \beta)  \\
  = \sum_{i\in \setI_0} c_i J(h;\delta_1^i\backslash b ,
  \delta_2^i\backslash b , \delta_3^i \backslash b,
  \delta_4^i\backslash b | \beta ) .
\end{multline} 
By~(\ref{finala}) and~(\ref{finalb}), we can then prove that 
$\{\delta_1 , \delta_2 \} = \{\delta_1^i , \delta_2^i \}$ and 
$\{\delta_3 , \delta_4 \} = \{\delta_3^i , \delta_4^i \}$. 
The hypothesis ${\cal H}(n+1)$ then holds.

Combining the two cases, we can thus conclude that the set of
inequalities in $\setIngleton$ is not redundant.  Since $\GammaIn$ is
full-dimensional, $\setIngleton$ is indeed the unique, non-redundant
set of Ingleton inequalities characterizing $\GammaIn$,
\cite{Schrijver2003}, p.64.

\section{Conclusion}
We have identified the unique minimal characterization for the set of
polymatroids satisfying all Ingleton inequalities.  Knowing this set
can greatly simplify computation of Ingleton-LP bounds for the
multicast capacity of linear network codes. Compared to na\"ive
enumeration of all Ingleton inequalities, approximately of the order
$16^n$, the actual number of necessary inequalities has size
approximately of the order $6^n/4 - 5^n$. The complexity reduction is
significant, in particular for large  $n$.

\section*{Acknowledgement}
This work was supported by the Australian Government under ARC grant
DP0557310. This work was performed while L. Guill\'e was visiting the
University of South Autralia.

\bibliographystyle{IEEEtran}


\end{document}